\documentclass[copyright,creativecommons]{eptcs}
\providecommand{\event}{E. Formenti (Ed.): AUTOMATA \& JAC 2012}

\usepackage{amsthm,enumerate,amsmath,amscd,amssymb}
\usepackage{comment}
\usepackage{rotating}
\usepackage{color}
\usepackage{xcolor}
\usepackage{colortbl}

\newtheorem{lemma}{Lemma}
\newtheorem{definition}{Definition}
\newtheorem{corollary}{Corollary}
\newtheorem{proposition}{Proposition}
\newtheorem{example}{Example}
\newtheorem{theorem}{Theorem}
\newtheorem{conjecture}{Conjecture}
\newtheorem{question}{Question}
\newtheorem{remark}{Remark}
\newtheorem{problem}{Problem}

\newcommand{\R}{\mathbb{R}}
\newcommand{\Q}{\mathbb{Q}}
\newcommand{\Z}{\mathbb{Z}}
\newcommand{\C}{\mathbb{C}}
\newcommand{\N}{\mathbb{N}}
\newcommand{\abs}[1]{\left\vert #1 \right\vert}
\newcommand{\si}[1]{\begin{sideways} #1 \end{sideways}}
\newcommand{\diff}[1]{{\color{red} #1}}
\newcommand{\same}[1]{{#1}}

\def\myd#1{\text{\d{\ensuremath#1}}}

\newcommand{\CA}{\mathtt{CA}}
\newcommand{\SUR}{\mathtt{SUR}}
\newcommand{\REV}{\mathtt{REV}}
\newcommand{\PRES}{\mathtt{PRES}}
\newcommand{\ID}{\mbox{id}}
\newcommand{\INF}{{}^\infty}

\title{Topology Inspired Problems for Cellular Automata, and a Counterexample in Topology\thanks{Research supported by the Academy of Finland Grant 131558}}
\def\titlerunning{Topology and Cellular Automata}

\author{
	Ville Salo
		\institute{TUCS -- Turku Center for Computer Science, \\
		University of Turku, Finland}
		\email{vosalo@utu.fi}
	\and
	Ilkka T\"orm\"a
		\institute{University of Turku, Finland}
		\email{iatorm@utu.fi}
}
\def\authorrunning{Ville Salo, Ilkka T\"orm\"a}

\begin{document}
\maketitle

\begin{abstract}
We consider two relatively natural topologizations of the set of all cellular automata on a fixed alphabet. The first turns out to be rather pathological, in that the countable space becomes neither first-countable nor sequential. Also, reversible automata form a closed set, while surjective ones are dense. The second topology, which is induced by a metric, is studied in more detail. Continuity of composition (under certain restrictions) and inversion, as well as closedness of the set of surjective automata, are proved, and some counterexamples are given. We then generalize this space, in the sense that every shift-invariant measure on the configuration space induces a pseudometric on cellular automata, and study the properties of these spaces. We also characterize the pseudometric spaces using the Besicovitch distance, and show a connection to the first (pathological) space.
\end{abstract}

\section{Introduction}

Cellular automata are a class of discrete dynamical systems. They are defined on the set of all two-directional infinite sequences of symbols from a finite alphabet, and the dynamics are given by a local function, which is synchronously applied at each coordinate. Cellular automata have been previously studied at least in the contexts of algorithmics, computability and dynamical systems theory \cite{Ka05}. The set of all cellular automata, or even a large subset of it, has been studied in relatively few articles. Usually, such works have been concerned with the algebraic structure of the automorphism group or endomorphism monoid of a subshift \cite{BoLiRu88,BoFrKi90,Ho10}.

In the article \cite{SaTo12}, we show that when the full shift is given a certain natural topology (the Besicovitch topology), cellular automata are exactly the continuous functions on it which compute the image of some cellular automaton at every point. In particular, the cellular automaton used cannot even be `changed' at, say, unary points of the full shift, where not much information about the cellular automaton is shown in the image. The crucial idea in the proof is topologizing the set of all cellular automata, although the choice of topology is rather arbitrary. This is the only use of a topology on the space of cellular automata we are aware of. Since the space is countable, one might assume that topology is necessarily useless. A good counterexample to this intuition is found for instance in \cite{Fu55}, where the infinitude of primes is proved using topology, in a rather beautiful way.

In hope of finding other uses for topologizing cellular automata, we introduce two naturally arising topologies for the set of cellular automata. One of these topologies, defined in terms of pointwise behavior, turns out to be rather pathological, in that it gives a countable space which is not first-countable. Some other such examples can be found in \cite{SeSt95}, but they are all based on a different idea, and have quite distinct properties from our space.

In addition to the naturally arising pathological example, we define a family of topologies which not only arise naturally, but also behave that way. The topologies are given by measures of the configuration space, and we call them difference set topologies. The topologies arise from measuring the difference of preimages of cylinders, and thus they measure the difference of two CA in one step of evolution. As one might guess, it turns out convergence in such a topology does not imply that the limit shares the dynamical properties of the elements of the converging sequence. However, many `one-step' properties play well with our topology: in the uniform Bernoulli case, which is the most natural representative of the family of topologies, we prove that surjective cellular automata are a closed set using the balance property, that inversion is continuous, and that composition is continuous under certain natural restrictions. Also, our proofs of perfectness and non-closedness of well-known subspaces use interesting approximation results: for instance, we turn the XOR automaton invertible with an arbitrarily small change in its local function. This was also one of the key (although also one of the simplest) ideas in the proof of \cite{Ka90}.

The paper is organized as follows. Section~\ref{sec:Defs} consists of the definitions and notation used in this paper. In Section~\ref{sec:Problems}, we briefly consider some topologies on cellular automata which seem promising at first, but fail to satisfy some natural constraints. In particular, we show that in a topology defined by pointwise behavior, composition is far from continuous, and discuss its other pathological properties. In Section~\ref{sec:Combinatorics}, we give a combinatorial description of the difference set topology given by the uniform Bernoulli measure, and present some of its properties with combinatorial proofs. In Section~\ref{sec:Measures}, we give the general definition of the difference set topologies for arbitrary shift-invariant measures. We study the properties of such topological spaces, and rephrase some of the combinatorial results in terms of measures. Section~\ref{sec:Future} consists of our conclusions and directions for future work.

\section{Definitions}
\label{sec:Defs}

Let $\Sigma$ be a finite set, called the \emph{state set} or \emph{alphabet}, which we assume to always have cardinality greater than $1$. The set $\Sigma^\Z$ of bi-infinite state sequences, or \emph{configurations}, is called the \emph{full shift on $\Sigma$}. If $x \in \Sigma^\Z$, then we denote by $x_i$ the $i$th coordinate of $x$, and we adopt the shorthand notation $x_{[i,j]} = x_i x_{i+1} \ldots x_j$.  A \emph{word} is an element of $\Sigma^n$ for some $n$, and we write $\Sigma^*$ for the set of all words. We define $N(r)$ for the interval $[-r, r]$, and call $w \in \Sigma^{N(r)}$ a \emph{centered word}, denoting $r(w) = r$. We use the notation $\Sigma^\leftrightarrow$ for the set of all centered words. The indexing notations $x_i$ and $x_{[i,j]}$ are extended to (centered) words, with $v_0$ being the central coordinate of a word $v \in \Sigma^\leftrightarrow$. For $v \in \Sigma^\leftrightarrow$ and $u, w \in \Sigma^*$ with $|u| = |w|$, we write $uvw$ for the centered word sharing its center with $v$, and $v_{[i, j]}$ for the word $v_i \cdots v_j$. For words $u,v,w \in \Sigma^*$, the notation ${}^\infty u \myd v w^\infty$ defines in a natural way a configuration which is $u$-periodic to the left and $w$-periodic to the right, with the word $v$ starting at coordinate $0$.

We define a metric $d_C$, called the Cantor metric, on the full shift by
\[ d_C(x,y) = \sum_{x_i \neq y_i} 2^{-|i|} \]
for all $x,y \in \Sigma^\Z$. This definition makes $\Sigma^\Z$ a compact metric space. We define the \emph{shift map} $\sigma : \Sigma^\Z \to \Sigma^\Z$ by $\sigma(x)_i = x_{i+1}$. For a word $w$ (centered or not) and $n \in \Z$, we define the \emph{cylinder of $w$ at $n$} by $[w]_n = \{ x \in \Sigma^\Z \;|\; x_{n+i} = w_i \mbox{ for all relevant } i \}$. If $W$ is a set of words, then $[W]_n = \bigcup_{w \in W} [w]_n$. If $W$ is finite, these sets are also called cylinders, and they form a clopen (closed and open) basis for the topology of $\Sigma^\Z$.

A \emph{cellular automaton} is a continuous function $c : \Sigma^\Z \to \Sigma^\Z$ with the property $c \circ \sigma = \sigma \circ c$. Alternatively, cellular automata are defined by \emph{local functions} $F : \Sigma^{N(r)} \to \Sigma$ for some $r \in \N$ such that $c(x)_i = F(x_{[i-r,i+r]})$. Any $r$ that can be chosen for the local function is called a \emph{radius} of $c$, and the minimal radius is denoted $r(c)$. We also denote $N(c) = N(r(c))$. A cellular automaton $c$, being continuous, defines a dynamical system $(\Sigma^\Z, c)$. The notation $\CA$ stands for the set of all CA on $\Sigma^\Z$. We denote by $\SUR$ and $\REV$ the sets of surjective and injective cellular automata, respectively. It is known (see \cite{He69}) that $\REV \subset \SUR$, and that bijective cellular automata are in fact reversible, that is, the inverse function is in $\CA$ as well. Also, surjective cellular automata are exactly those $c \in \CA$ that exhibit the so-called \emph{balance property} that $|c^{-1}(w)| = |\Sigma|^{2r(c)}$ for all $w \in \Sigma^*$ with $|w| \geq 1$, when $c$ is considered as a word function $c : \Sigma^{|w|+2r(c)} \to \Sigma^{|w|}$.

For $u, v \in \Sigma^\leftrightarrow$, we denote $u \prec v$ if $v = sut$ for some $s,t \in \Sigma^*$, and $u \sim v$ if $u_0 = v_0$. We always assume $\tau$ is a permutation of $\Sigma$ without fixed points. When $\Sigma = \{0, 1\}$, the permutation $a \mapsto \bar{a}$ is defined as $(0\;1)$ in cycle notation.

We also need some notions for topological spaces and their subsets, mostly in Section~\ref{sec:Problems}. A subset of a topological space is \emph{sequentially closed}, if it contains the limits of its converging sequences. A topological space $X$ is said to be
\begin{itemize}
\item \emph{metrizable}, if the topology can be induced by a metric,
\item \emph{normal}, if any two disjoint closed subsets of $X$ can be separated with open neighborhoods,
\item \emph{completely regular}, if any closed set $F$ of $X$ and a point $x \notin F$ can be separated by a continuous function $f : X \to [0, 1]$ in the sense that $f(x) = 0, f(F) = \{1\}$,
\item \emph{sequential}, if its sequentially closed subsets of $X$ are exactly the closed sets,
\item \emph{of first category}, if it is a countable union of nowhere dense sets,
\item \emph{totally disconnected}, if its connected components are singletons,
\item \emph{zero dimensional}, if it has a clopen base,
\item \emph{first-countable}, if every point has a countable neighborhood basis,
\item \emph{second-countable}, if the topology has a countable base.
\end{itemize}
A function $f : X \to Y$ between metric spaces $(X,d)$ and $(Y,d')$ is \emph{Lipschitz} if $d'(f(x),f(y)) < C \cdot d(x,y)$ for some constant $C > 0$. In particular, Lipschitz functions are continuous.

In this article, a \emph{measure} is always understood as a probability measure on the Borel subsets of $\Sigma^\Z$. A measure $\mu$ is \emph{shift invariant} if $\mu(C) = \mu(\sigma(C))$ for all Borel sets $C \subseteq \Sigma^\Z$. It is known (see \cite{DeGrSi76}) that all measures are \emph{regular}, meaning that for every Borel set $C \subset \Sigma^\Z$ we have $\mu(C) = \inf \{ \mu(U) \;|\; C \subseteq U, U \mbox{ open} \}$. For two measures $\mu$ and $\nu$, we denote $\mu \gg \nu$ if $\mu(C) = 0$ implies $\nu(C) = 0$.

A shift invariant measure $\mu$ is \emph{ergodic}, if $\sigma(C) = C$ implies $\mu(C) \in \{0,1\}$. Examples of ergodic measures are the \emph{Bernoulli measures} $\mu_p$, given by a map $p : A \to [0,1]$ with $\sum_{a \in A} p(a) = 1$, which are defined by setting $\mu_p([w]_0) = \prod_{i=0}^{|w|-1} p(w_i)$. If $\mu$ is ergodic, $\nu$ is shift invariant and $\mu \gg \nu$, then $\mu=\nu$. The famous \emph{Birkhoff's pointwise ergodic theorem} states that if $\mu$ is a shift invariant measure and $f : \Sigma^\Z \to \C$ a $\mu$-integrable function, then
\[ f^*(x) = \lim_{n \rightarrow \infty} \frac{1}{2n+1} \sum_{i=-n}^n f(\sigma^i(x)) \]
is defined for $\mu$-almost all $x$. Also, the function $f^*$ is $\mu$-integrable and satisfies
\[ \int_{\Sigma^\Z} f^* \; d \mu = \int_{\Sigma^\Z} f \; d \mu. \]

\section{Problematic Topologies}
\label{sec:Problems}

When topologizing a set of objects, one usually wishes the topology to be somehow related to the nature of these objects, or some structure given to the set itself. In a topological group, for example, the group operations $(g,h) \mapsto g \cdot h$ and $g \mapsto g^{-1}$ are required to be continuous. In the case of cellular automata, a reasonable requirement is that \emph{composition should be continuous}, at least with some natural restrictions. In order to avoid trivialities, we also require that \emph{the space should be nondiscrete and Hausdorff}. In this section, we briefly discuss some topologies that fail to fulfill these properties, also noting a possibly interesting counterexample in topology.

\begin{definition}
The \emph{pointwise topology} on $\CA$ has the sets
\[ U_x(a) = \{c \in \CA \;|\; c(x)_0 = a\}, \]
indexed by $x \in \Sigma^\Z$ and $a \in \Sigma$, as a subbase. For $x \in \Sigma^\Z, c \in \CA$, we denote $U_x(c) = U_x(c(x)_0)$.
\end{definition}

This means that a sequence $(c_i)_{i \in \N}$ converges to $c \in \CA$ iff $(c_i(x))_{i \in \N}$ converges to $c(x)$ for all $x \in \Sigma^\Z$. Equivalently, $c_i \rightarrow c$ if and only if for all $x \in \Sigma^\Z$, we have $(c_i(x))_0 = c(x)_0$ for all large enough $i$. As we will see later, this does not characterize the pointwise topology. In the case $|\Sigma| = 2$, this space is a kind of `dual' of the usual Cantor space: The Cantor topology is given on the set of infinite sequences by taking the cylinders as a clopen base. In the pointwise topology, we take the cylinders as points (a CA corresponding to the cylinder it maps to $1$), and take as a subbase the sets of cylinders containing, or not containing, a given point.

The pointwise topology might seem like a very natural choice for topologizing the cellular automata, but already the composition operation fails to be continuous:

\begin{example}
Composition is not sequentially continuous (and thus not continuous) in the pointwise topology, even in $\REV$ with the alphabet $\{0,1\}$. Namely, let $i \in \N$, and consider the automata $c_i$ and $d_i$ that behave as the identity, except that $c_i$ transforms every pattern of the form $110^ia0^i11$ to $110^i\bar a0^i11$, and $d_i$ transforms every $0^{i+1}aa0^i10^iaa0^{i+1}$ to $0^{i+1}\bar a \bar a0^i10^i\bar a \bar a0^{i+1}$. It is routine to check that these automata are reversible (as they have period $2$), and that in the pointwise topology, both $(c_i)$ and $(d_i)$ converge to the identity automaton. However, $(c_i \circ d_i)$ does not, since the central cell of $c_i(d_i(\INF 0 \myd 1 0 \INF))$ is always $0$.
\end{example}

Of course, composition from either side with a constant cellular automaton is continuous.

The pointwise topology has some interesting properties: the reversible cellular automata form a closed set, and the set of surjective automata is dense. Note that since we have not proved that this space is sequential (and in fact it is not: see Theorem~\ref{thm:PointwiseProperties}), we cannot prove these propositions using sequences.

\begin{proposition}
\label{prop:REVYesClosed}
The set $\REV$ is closed in the pointwise topology.
\end{proposition}

\begin{proof}
Let $c$ be a point in $\CA - \REV$. Since a cellular automaton is reversible if and only if it is reversible on periodic points, we have that $c(x) = c(y)$ for some $p$-periodic points $x$ and $y$. Now, let $X_p = \{z \;|\; \sigma^p(z) = z\}$, and consider the set $U = \bigcap_{z \in X_p} U_z(c)$, which is an open neighborhood of $c$. Then, $d \in U \implies d(x) = d(y)$, and thus $U \subset \CA - \REV$. 
\end{proof}

\begin{proposition}
\label{prop:SURDense}
The set $\SUR$ is dense in the pointwise topology.
\end{proposition}

\begin{proof}
Let $c \in \CA$ be arbitrary, and let $U = \bigcap_{x \in X} U_x(c)$ be a neighborhood of $c$, where $X \subset \Sigma^\Z$ is finite. Let $n \in \N$ be such that every pair $x \neq y \in X$ differs in some coordinate in $N(n)$. We assume that $\Sigma = \{0, \ldots, k-1\}$, so that we can use addition modulo $k$ on $\Sigma$. For a centered word $w \in \Sigma^{N(n)}$, we define $X(w) \in \Sigma$ as the letter $x_{n+1}$, if $x_{N(n)} = w$ for some $x \in X$, and $0$ otherwise. Note that $X(w)$ is well-defined by the choice of $n$.

We define the automaton $d \in \CA$, which has radius $n + 1$, by
\[ d(x)_0 = c(x)_0 + x_{n+1} - X(x_{N(n)}) \bmod k \]
for all $x \in \Sigma^\Z$. We immediately see that $d(x)_0 = c(x)_0$ for all $x \in X$, so that $d \in U$. Also, $d$ is right-permutive, meaning that $d(v) \neq d(w)$ holds whenever $v,w \in \Sigma^{N(n+1)}$ differ only in their rightmost coordinate $n+1$. It is clear that such automata are always surjective, so the claim is proved. 
\end{proof}

We note that whether a set $F \subseteq \CA$ is dense in another set $G \subseteq \CA$ in the pointwise topology has a natural interpretation even outside our formalism: this is the case if and only if whenever one chooses an automaton $c \in G$ and a finite set of points $X \subseteq \Sigma^\Z$, some automaton $d \in F$ satisfies $d(x)_0 = c(x)_0$ for all $x \in X$.

The pointwise topology has the curious property that, while the underlying set is countable, the space is not first-countable. To our knowledge, examples of such spaces are nontrivial to construct. Of the $143$ examples in \cite{SeSt95}, only $5$ spaces are countable but not first-countable, and we compare them with the pointwise $\CA$ space in Table~\ref{tab:Topo}.

\begin{theorem}
\label{thm:PointwiseProperties}
The space $\CA$ with the pointwise topology has the following properties:
\begin{itemize}
\item It is countable.
\item It is Hausdorff.
\item It is not compact.
\item It is perfect.
\item It is totally disconnected.
\item It is of first category.
\item It is not first-countable, and in fact no point has a countable neighborhood basis. It is thus neither metrizable nor second-countable.
\item It is normal.
\item It is not sequential.
\end{itemize}
\end{theorem}

\begin{proof}
Countability and the Hausdorff property are trivial. The lack of compactness follows from $(U_x(a))_{x \in S^\Z, a \in \Sigma}$ not having a finite subcover. The space is easily seen to be perfect, since the finite intersections of $U_x(a)$ form a base, and no such set is a singleton. The space is totally disconnected because every $U_x(a)$ is clopen (its complement is $\bigcup_{b \in \Sigma - \{a\}} U_x(b)$), and it is of first category since every singleton set is closed and the space is countable.

As for lack of first-countability, let $c \in \CA$, and assume $(V_i)_{i \in \N}$ is a neighborhood basis for $c$. Since $(U_x(c))_{x \in \Sigma^\Z}$ forms a neighborhood subbasis for $c$ as well, there must exist finite sets $X_i \subset \Sigma^\Z$ such that $\bigcap_{x \in X_i} U_x(c) \subset V_i$ for all $i$. If we let $X = \bigcup_{i \in \N} X_i$, then clearly $(U_x(c))_{x \in X}$ is already a countable subbase for $c$. Let $y \notin X$, and for each $i$, let $c_i$ be the cellular automaton with $r(c_i) = i$ which behaves like $c$, except for mapping $y_{[-i, i]}$ to $\tau(c(y)_i)$. Clearly, $c_i \not\rightarrow c$ since $c_i(y)_0 \neq c(y)_0$ for all $i$. However, when $x \neq y$, we have $c_i \in U_x(c)$ for large enough $i$, and thus $c_i \rightarrow c$, a contradiction.

Finally, we prove the normality. Let $c_i, d_i \in \CA$ for all $i \in \N$ be such that the sets $X = \{c_i \;|\; i \in \N\}$ and $Y = \{d_i \;|\; i \in \N\}$ are closed and disjoint. Since the space is countable, this covers all the closed sets. Inductively for all $i \in \N$, we choose clopen neighborhoods $U_i \ni c_1,\ldots,c_i$ and $V_i \ni d_1,\ldots,d_i$ such that $U_i \cap Y = V_i \cap X = U_i \cap V_i = \emptyset$ as follows. First, $U_0 = V_0 = \emptyset$. Let $U_{i-1}$ and $V_{i-1}$ be defined. Let $U$ be a clopen neighborhood of $c_i$ such that $U \cap Y = \emptyset$, and define $U_i = U_{i-1} \cup (U - V_{i-1})$. Similarly, $V_i = V_{i-1} \cup (V - U_i)$ for a clopen neighborhood $V$ of $d_i$ with $V \cap X = \emptyset$. Now, $\bigcup_i U_i$ and $\bigcup_i V_i$ form disjoint open neighborhoods of $X$ and $Y$, respectively, and this proves normality.

The lack of sequentiality is seen as follows: The set $\SUR$ is dense by Proposition~\ref{prop:SURDense}, and thus not closed, since now $\SUR \subsetneq \CA = \overline{\SUR}$. However, we will prove in Theorem~\ref{thm:SURClosed} that $\SUR$ is closed in another topology we define in Section~\ref{sec:Combinatorics}, and Lemma~\ref{prop:Finer} shows that such sets are sequentially closed in the pointwise topology. These results do not depend on this theorem. 
\end{proof}

\begin{remark}
In fact, the above proof for normality works as such in any countable space with a clopen base. Without countability, we only know that such a space is completely regular.
\end{remark}

We now note a further connection with the Cantor space, namely that the pointwise topology is zero-dimensional and perfect. It is known that a compact, perfect, zero-dimensional and metrizable space is homeomorphic to the Cantor space. We do not know whether a similar characterization exists for the pointwise topology, and whether the `duality' can be formalized.

\begin{problem}
Characterize the pointwise topology on $\CA$ using some of its topological properties.
\end{problem}

\begin{table}
\centering
\begin{tabular}{l|llllll|}
\cline{2-7}
																	& \si{normal Hausdorff}	& \si{compact}& \si{sequential}	& \si{tot. disconnected}	& \si{first category} &	\si{perfect}	\\
\hline
\multicolumn{1}{|c|}
{$\CA$ with pointwise topology}		& yes										& no					& no							& yes											&	yes									& yes						\\
\arrayrulecolor{black!30}
\hline
\arrayrulecolor{black}
\multicolumn{1}{|c|}
{26. Arens-Fort space}						& \same{yes}						& \same{no}		& \same{no}				& \same{yes}							& \diff{no}						& \diff{no}			\\
\multicolumn{1}{|c|}
{35. $\Q^*$}											& \diff{no}							& \diff{yes}	& \diff{yes}			& \diff{no}								&	\same{yes}					& \same{yes}		\\
\multicolumn{1}{|c|}
{98. Appert space}								& \same{yes}						& \same{no}		& \same{no}				& \same{yes}							&	\diff{no}						& \diff{no}			\\
\multicolumn{1}{|c|}
{99. maximal compact topology}		& \diff{no}							& \diff{yes}	& \diff{yes}			& \same{yes}							&	\diff{no}						& \diff{no}			\\
\multicolumn{1}{|c|}
{114. single ultrafilter topology}& \same{yes}						& \same{no}		& \same{no}				& \same{yes}							&	\diff{no}						& \diff{no}			\\
\hline
\end{tabular}
\caption{Properties of the countable but not first-countable counterexamples of \cite{SeSt95} contrasted with our space. The space $\Q^*$ is the one-point compactification of $\Q$. In the table, the spaces which are not normal Hausdorff are neither normal nor Hausdorff. The space $\Q^*$ is in fact connected. The spaces which are not perfect do not even contain induced perfect subspaces.}
\label{tab:Topo}
\end{table}

Finally, we note that one might also define a topology on $\CA$ in one of the following ways: $c_i \rightarrow c$ if for all $x \in \Sigma^\Z$, we have $c_i(x) = c(x)$ for all large enough $i$, or by defining that $c_i \rightarrow c$ if for all $w \in \Sigma^\leftrightarrow$, we have $c_i(x)_0 = c(x)_0$ for all $x \in [w]_0$ for large enough $i$. However, both of these attempts result in the discrete topology, as is easily verified.

\section{Topologizing the CA by measuring preimages: the combinatorial approach}
\label{sec:Combinatorics}

\begin{definition} \label{delta}
Let $c,d \in \CA$, and $r$ a common radius for $c$ and $d$. Then the \emph{difference set} of $c$ and $d$ for the radius $r$ is
\[ D^c_d = \{ w \in \Sigma^{N(r)} \;|\; c(w) \not\sim d(w) \}, \]
where the CA are taken as word functions $c : \Sigma^{N(r)} \to \Sigma^{N(r-r(c))}$ and $d : \Sigma^{N(r)} \to \Sigma^{N(r-r(d))}$. The \emph{distance} between $c$ and $d$ is
\[ \delta(c,d) = \frac{\abs{D^c_d}}{\abs{\Sigma}^{2r+1}}. \]
We call the space $\CA$ with this metric the \emph{uniform Bernoulli space}.
\end{definition}

The name `uniform Bernoulli space' comes from the fact that we later define pseudometrics on $\CA$ using shift-invariant measures, and the pseudometric $\delta$ is given by the uniform Bernoulli measure. The value $\delta(c,d)$ above does not depend on the exact value of $r$, and it is left implicit in the notation $D^c_d$. For $c \circ d$, we always use the radius $r(c) + r(d)$.

\begin{example}
\label{ex:NotPrecompact}
The space $\CA$ is not totally bounded (and thus not compact): Let $\{c_i \in \CA \;|\; i=1, \ldots, n\}$ be a finite set of automata, and let $r = \max \{r(c_i) \;|\; i=1, \ldots, n\}$. Consider the $r+1$-shift automaton $\sigma^{r+1}$. Given a word $w \in \Sigma^{N(r+1)}$, we have $\sigma^{r+1}(w) \sim c_i(w)$ iff $w_{r+1} = c_i(w_0)$. Thus
\[ \delta(\sigma^{r+1},c_i) = 1 - \abs{\Sigma}^{-1}, \]
and $\CA$ cannot be covered by a finite number of open $(1 - \abs{\Sigma}^{-1})$-balls.
\end{example}

\begin{example}
The function $(\circ) : \CA^2 \to \CA$ is not continuous. In fact, even the function $(d \mapsto d \circ c)$ is not continuous for all $c \in \CA$: Let $\Sigma = \{0, 1\}$ and let $c_q$ for $q \in \{0,1\}$ be the all-$q$ CA defined by $c_q(x)_i = q$ for all $x \in \Sigma^\Z, i \in \Z$. We let
\[
d_i(x)_j = \left\{ \begin{array}{ll}
  1, & \mbox{if } x_{[j-i, j+i]} = 0^{2i+1} \\
  0, & \mbox{otherwise.}
\end{array} \right.
\]
Clearly $d_i \rightarrow c_0$, but since $d_i \circ c_0 = c_1$ for all $i$, $d_i \circ c_0$ does not approach $c_0 \circ c_0 = c_0$.
\end{example}

Here, we used a rather pathological CA as the rightmost argument, and in fact, with suitable restrictions, we do obtain continuity.

\begin{theorem}
The restriction of $(\circ)$ to $\CA \times \SUR$ is continuous.
\end{theorem}

\begin{proof}
Since $\CA \times \SUR$ is a metric space, it is enough to check that limits of sequences commute with $(\circ)$. So let $c_i, c \in \CA$ and $d_i, d \in \SUR$ with $(c_i, d_i) \rightarrow (c, d)$. We need to show $c_i \circ d_i \rightarrow c \circ d$. Let $i$ be large enough that $\delta(c_i,c) < \epsilon$ and $\delta(d_i,d) < \epsilon$, and assume that $r(c_i) \geq r(c)$ and $r(d_i) \geq r(d)$. We need an upper bound on the size of $D^{c_i \circ d_i}_{c \circ d}$.

First, let us give an upper bound for the size of the set $A$ of all words $w \in \Sigma^{N(r(c_i) + r(d_i))}$ such that $d_i(w)_j \neq d(w)_j$ for some $j \in N(c)$. If $w \in A$, then at least one of $w_{N(d_i)+j}$ for $j \in N(c)$ must be in $D^{d_i}_d$. Thus
\[ |A| \leq |N(c)| \cdot \abs{D^{d_i}_d} \cdot |\Sigma|^{2r(c_i)}. \]

Then, let us give an upper bound for the size of $B = D^{c_i \circ d_i}_{c \circ d} - A$. Let $w \in B$, so that $d_i(w)_{N(c)} = d(w)_{N(c)}$. Since $w \in D^{c_i \circ d_i}_{c \circ d}$, we have $c_i(d_i(w)) \not\sim c(d(w)) \sim c(d_i(w))$, which implies $c_i(d_i(w)) \not\sim c(d_i(w))$, and thus $d_i(w) \in D^{c_i}_c$. Now we have that $d_i(B) \subset D^{c_i}_c$, and since $d_i \in \SUR$, the balance property implies
\[ |B| \leq \abs{D^{c_i}_c} \cdot |\Sigma|^{2r(d_i)}. \]

Of course, $\abs{D^{c_i \circ d_i}_{c \circ d}} = |A| + |B|$, which implies $\delta(c_i \circ d_i, c \circ d) < (2r(c) + 2)\epsilon$ by a direct calculation. 
\end{proof}

The following lemma is proved similarly as the previous one, and shows that composition on the left is always continuous.

\begin{lemma} \label{rcomp}
If $e \in \CA$, then the function $c \mapsto e \circ c$ is (Lipschitz-)continuous in $\CA$. More specifically,
\[ \delta(e \circ c, e \circ d) \leq (2r(e)+1) \cdot \delta(c, d) \]
for all $c, d \in \CA$.
\end{lemma}

\begin{proof}
Let $c, d \in \CA$, and let $r=\max(r(c),r(d))$. Now if $e(c(w)) \not\sim e(d(w))$ for some word $w \in \Sigma^{N(r+r(e))}$, then necessarily $c(w)_{N(e)} \neq d(w)_{N(e)}$. Thus we have
\begin{eqnarray*}
\abs{D^{e \circ c}_{e \circ d}} & \leq & \abs{ \{ w \in \Sigma^{N(r+r(e))} \;|\; c(w)_{N(e)} \neq d(w)_{N(e)} \} } \\
& \leq & \abs{N(e)} \cdot \abs{D^c_d} \cdot \abs{\Sigma}^{2r(e)}.
\end{eqnarray*}
The claim again follows by a direct calculation. 
\end{proof}

Using the above lemma, we prove that inversion is also continuous.

\begin{theorem}
The function $(\cdot)^{-1} : \REV \to \REV$ is continuous.
\end{theorem}

\begin{proof}
Let $c \in \REV$. We prove that $(\cdot)^{-1}$ is continuous at $c$. For that, let $d \in \REV$ such that $r=r(d)>r(c)$ and $r'=r(d^{-1})>r(c^{-1})$.

Let $w \in \Sigma^{N(r+r')}$ be such that $d(w) \in D^{c^{-1}}_{d^{-1}}$. Then we have $c^{-1}(d(w)) \not\sim d^{-1}(d(w)) \sim c^{-1}(c(w))$, which implies $w \in D^{c^{-1} \circ c}_{c^{-1} \circ d}$. This means that the $d$-preimage of $D^{c^{-1}}_{d^{-1}}$ is included in $D^{c^{-1} \circ c}_{c^{-1} \circ d}$.

Since $d \in \SUR$, we again have $\abs{D^{c^{-1}}_{d^{-1}}} \cdot \abs{\Sigma}^{2r+1} \leq \abs{D^{c^{-1} \circ c}_{c^{-1} \circ d}}$ by the balance property. An application of Lemma~\ref{rcomp} now yields
\[ \delta(c^{-1},d^{-1}) \leq \delta(c^{-1} \circ c, c^{-1} \circ d) \leq (2r(c^{-1})+1) \cdot \delta(c,d), \]
which proves the claim. 
\end{proof}

The space has no isolated points, and in fact we obtain a large amount of approximation results by using the continuity of $(\circ)$.

\begin{proposition}
All of $\CA$, $\SUR$ and $\REV$ are perfect as topological spaces.
\end{proposition}

\begin{proof}
Given CA $c$, we construct the CA $c_i$ with $r(c_i) = r(c) + i$, which are equivalent to $c$ on all words of $\Sigma^{N(c_i)}$, except $0^{N(c_i)}$, on which they disagree with $c$. Clearly $c_i \rightarrow c$, but $c_i \neq c$ holds for all $i$, and thus $\CA$ is a perfect space.

Now consider the cellular automata $c_i$ which function as identity maps, except for mapping the central cell of $110^ia0^i11$ to $\tau(a)$. It is clear that $c_i \in \REV$ for all $i$, since these automata keep the occurrences of $11$ untouched, and can only change the central cell between two occurrences. It is also clear that they converge to the identity map.

Now consider an arbitrary surjective CA $d$. All of the $c_i \circ d$ are distinct, since $d$ is surjective and the $c_i$ are distinct, and by the continuity of $(\circ)$ restricted to $\CA \times \SUR$, the sequence $c_i \circ d$ converges to $\ID \circ d = d$. If $d$ is reversible, then so are all $c_i \circ d$. We have obtained that also $\SUR$ and $\REV$ are perfect as topological spaces. 
\end{proof}

\begin{theorem}
\label{thm:SURClosed}
$\SUR$ is closed.
\end{theorem}

\begin{proof}
Let $c \notin \SUR$ and $d \in \SUR$. We show that $\delta(c, d)$ is bounded from below by a positive constant depending only on $c$. That is, we will prove there exists $\epsilon > 0$ such that $B_\epsilon(c) \cap \SUR = \emptyset$.

We may assume $r(d) \geq r(c)$. Since $c$ is not surjective, there exists a word $v \in \Sigma^{N(n)}$ for some $n$ such that
\[ \abs{c^{-1}(v)} \geq \abs{\Sigma}^{2r(c)} + 1. \]
Now let $w \in \Sigma^{N(r(d)+n)}$ be such that $c(w)_{r(d)} = v$. We either have $d(w) = v$, or $w_{N(d)+i} \in D^c_d$ for some $i \in N(n)$. In terms of cardinalities this implies
\begin{align*}
& \quad \abs{d^{-1}(v)} + (2n+1) \cdot \abs{D^c_d} \cdot \abs{\Sigma}^{2n} \\
&\geq \abs{c^{-1}(v)} \cdot (\abs{\Sigma}^{2(r(d) - r(c))}) \\
&\geq (|\Sigma|^{2r(c)} + 1)(|\Sigma|^{2(r(d) - r(c))}),
\end{align*}
and since $\abs{d^{-1}(v)} = \abs{\Sigma}^{2r(d)}$ by the balance property, it follows that $\abs{D^d_c} \geq (2n+1)^{-1}\cdot\abs{\Sigma}^{2(r(d) - r(c) - n)}$. But then $\delta(c, d) \geq n^{-1}\cdot|\Sigma|^{-2(r(c)+n)}$, and the claim is proved. 
\end{proof}

In fact, by a result in \cite{KaVaZe09}, one can choose $n \leq \abs{\Sigma}^{2r(c)}$ in the previous proof, yielding a bound
\[ \delta(c, d) \geq \abs{\Sigma}^{-2(r(c)-\abs{\Sigma}^{2r(c)})}. \]

Like surjectivity and inversion, commuting with a fixed cellular automaton is another intuitively `one-step' property, and thus should behave well in our topology. Using the continuity of $\circ$, we indeed verify this intuition, and prove that the limit of CA commuting with a given surjective CA also commutes with the CA.

\begin{example}
The commutator of a surjective CA $c$ is closed. Consider the function $d \mapsto \delta(c \circ d, d \circ c)$. Since composition from both left and right by a surjective cellular automaton is continuous and the metric is a continuous function $\CA^2 \to \R$, we obtain that the preimage of $0$ is a closed set. But obviously this is just the commutator.
\end{example}

\begin{example}
For each $p \geq 1$, the set of $p$-periodic CA is closed. Namely, if $c_i \rightarrow c$ and $c_i^p = \ID$ for all $i$, then by the continuity of composition on $\REV$ we have $c^p = \ID$.
\end{example}

As opposed to the above example and the case of $\SUR$, the set $\REV$ of reversible cellular automata is not a closed subset of the uniform Bernoulli space.

\begin{example}
\label{ex:REVNotClosed}
The set $\REV$ is not closed if $|\Sigma| \geq 3$. Let $c$ be the CA
\[
c(x)_j = \left\{ \begin{array}{ll}
  x_j + x_{j-1} \bmod 2, & \mbox{if } x_j \in \{0, 1\} \\
  2, & \mbox{otherwise}
\end{array} \right.
\]
That is, $c$ is the XOR-with-left-neighbor automaton, which addionally fixes $2$s and interprets them as $0$s. We let $c_i$ behave as $c$, except for fixing the current cell if they see no $2$'s in the neighborhood $N(i)$. Since $\frac{2^{\abs{N(i)}}}{3^{\abs{N(i)}}} \longrightarrow 0$, we see that $c_i \longrightarrow c$.

Of course, $c$ is not in $\REV$, since $^\infty 0 ^\infty$ and $^\infty 1 ^\infty$ have the same image. All of $c_i$, however, are in $\REV$: Let $c_i(x) = c_i(y)$ for some $x \not= y$. Since $2$'s are never created or destroyed, $x$ and $y$ have $2$'s in the same coordinates, and thus for both $x$ and $y$, the set of coordinates that are fixed when applying XOR are equal. In every point, a fixed coordinate must occur at least every $|N(i)|$ steps, since if no $2$'s occur in a block of size $|N(i)|$, the middle cell is fixed. But these two facts clearly imply $x$ and $y$ are equal, which proves the claim.
\end{example}

We show that the uniform Bernoulli space is not complete either, which strengthens the intuition given in Example~\ref{ex:NotPrecompact} that sequences of cellular automata quite rarely converge to anything (as one might expect).

\begin{example}
The space $\CA$ is not complete. We will show this by constructing a Cauchy sequence $(c_i)$ in $\CA$ without a limit. Start with the identity CA $c_1$ and for each $i > 1$, define $c_i$ recursively as follows: Let $(n_i)$ be a sequence in $\N$ with $n_i \geq i$ and the property that
\begin{equation}
\label{eq:d-property}
\sum_{i=1}^{k-1} \abs{\Sigma}^{-n_{i+j}} \leq \frac{1}{2} \cdot \abs{\Sigma}^{-\sum_{i=1}^j n_i}
\end{equation}
for all $j,k \in \N$ with $j<k$. This is the case, if $(n_i)$ grows fast enough.

Let $r_i = r(c_i) = \sum_{k=1}^i n_k$, and for all words $w \in \Sigma^{N(r_i)}$, define
\[ c_i(w) = \left\{ \begin{array}{ll} \tau(c_{i-1}(w)_0), & \mbox{if } w_{[r(c_i)-n_i+1,r(c_i)]} = 0^{n_i} \\ c_{i-1}(w)_0, & \mbox{otherwise.} \end{array} \right. \]
Then $\delta(c_i,c_{i+1})=\abs{\Sigma}^{-n_i} \leq \abs{\Sigma}^{-i}$ for all $i$, and by the triangle inequality
\[ \delta(c_i,c_{i+k}) \leq \sum_{j=0}^{k-1} \delta(c_{i+j},c_{i+j+1}) \leq \abs{\Sigma}^{-i}\sum_{j=0}^{k-1} \abs{\Sigma}^{-j} < \frac{\abs{\Sigma}^{-i}}{1-\abs{\Sigma}^{-1}}, \]
so the sequence $(c_i)$ is Cauchy.

Moreover, for each $i,k \in \N$ with $i<k$, equation \eqref{eq:d-property} implies that
\begin{align*}
\delta(c_i,c_{i+k}) &\geq \delta(c_i,c_{i+1}) - \sum_{j=1}^{k-1} \delta(c_{i+j},c_{i+j+1}) \\
 &= \abs{\Sigma}^{-n_i} - \sum_{j=1}^{k-1} \abs{\Sigma}^{-n_{i+j}} \geq  \frac{1}{2} \cdot \abs{\Sigma}^{-n_i},
\end{align*}
so none of the terms $c_i$ can be a limit for the sequence.

Consider then an arbitrary CA $c$ not in the sequence $(c_i)$. We may choose it to have the radius $r(c_i)=\sum_{k=1}^i n_k$ for some large enough $i$. We show that $c$ is not a limit for $(c_i)$. First, note that $\delta(c,c_i) \geq \abs{\Sigma}^{-r(c_i)}$ holds. But now equation \eqref{eq:d-property} implies that
\begin{align*}
\delta(c,c_{i+k}) &\geq \delta(c,c_i) - \sum_{j=0}^{k-1} \delta(c_{i+j},c_{i+j+1}) \\
 &\geq \abs{\Sigma}^{-r(c_i)} - \sum_{j=1}^{k-1} \abs{\Sigma}^{-n_{i+j}} \geq  \frac{1}{2} \cdot \abs{\Sigma}^{-r(c_i)}
\end{align*}
for all $k>i$. Thus $c$ is not a limit for $(c_i)$.
\end{example}

Finally, let us briefly discuss the homomorphisms of \cite{Ka96} and \cite{BoLiRu88}, and some dynamical notions, in the uniform Bernoulli space and the pointwise topology.

\begin{definition}
Let $r$ be a radius of $c \in \REV$, and let
\[ R_c = \{(x_{[0, 2r-1]}, c(x)_{[-r, r-1]}) \;|\; x \in \Sigma^\Z\}, \]
the set of right stairs of $c$. We define $h_+(c) = \frac{|R_c|}{|\Sigma|^{3r}} \in \Q$.
\end{definition}

The map $h_+ : \REV \to \Q$ is in fact a group homomorphism from $(\REV, \circ)$ to $(\Q_{>0}, \cdot)$, and it compares the information flows to the left and to the right in the evolution given by an automaton \cite{Ka96}. We show by example that this function is not continuous, even though its definition is concerned with only one step.

\begin{example}
The morphism $h_+$ is not continuous at least with the alphabet $\Sigma = \{0, 1\}\times\{0, 1\}$. First, it is easy to see that for all $i$, there exists a length $k_i$ and a finite set of words $W_i \subset \Sigma^{k_i}$ with $|W_i| < \frac{1}{i}|\Sigma|^{k_i}$, such that $W_i$ defines an empty SFT when taken as the set of forbidden patterns. For $y \in \{0, 1\}^\Z$, define $J_i(y) = \{j \in \Z \;|\; y_{[j,j+k_i-1]} \in W_i\}$. Clearly, there must exist $n$ such that $\Z - J_i(y)$ does not contain an interval of length $n$ for any $y$, that is, coordinates $j$ such that $y_{[j,j+k_i-1]} \in W_i$ must occur with bounded gaps. For $j \in J_i(y)$, we define $P_i(j, y)$ as largest $j' \in J_i(y)$ such that $j' < j$.

We now define
\[ c_i(x)_j = \left\{\begin{array}{ll}
x_j, & \mbox{if } j \notin J_i(\pi_1(x)) \\
(\pi_1(x)_j, \pi_2(x)_{j'}), & \mbox{if } j \in J_i(\pi_1(x)),
\end{array}\right.\]
where $j' = P_i(j, \pi_1(x))$. That is, we shift information to the right at the positions marked by $W_i$.

From $\frac{|W_i|}{|\Sigma|^{k_i}} \rightarrow 0$, it clearly follows that $c_i \rightarrow \ID$. It is also easy to see that $h_+(c_i) = 2$ for all $i$.
\end{example}

We note that there is a more well-known homomorphism defined on $(\REV, \circ)$ called the \emph{gyration function}, defined in \cite{BoLiRu88}. Since this function is based on behavior on periodic points, it is very easy to find a counterexample for continuity in the uniform Bernoulli space, and a proof of continuity in the pointwise topology, when the codomain $\prod_{n = 1}^\infty \Z/n\Z$ is given the product topology.

The previous example also shows that sensitive cellular automata are not a closed set, as each $c_i$ is sensitive, but their limit is not.

\begin{example}
If $c \in \CA$, define the \emph{entropy} of $c$ as
\[ h(c) = \lim_{r \rightarrow \infty} \lim_{t \rightarrow \infty} \frac{\log N_c(r,t)}{t}, \]
where $N_c(r,t)$ is the number of different $r \times t$-rectangles in all spacetime diagrams of $c$. The entropy function is not continuous in general, even in $\REV$: Let $\Sigma=\{0,1\}^2$, and define the CA $c_i$ as the identity CA, except that maximal patterns of the form
\[ {}^{a_1}_{b_1}\left({}^0_0\right)^i{}^1_1{}^{a_2}_{b_2}\left({}^0_0\right)^i{}^1_1 \cdots \left({}^0_0\right)^i{}^1_1{}^{a_n}_{b_n}, \]
where $n \geq 2$ and the $a_j, b_j \in \Sigma$ are not part of any word $\left({}^0_0\right)^i{}^1_1$, are rotated to
\[ {}^{b_1}_{b_2}\left({}^0_0\right)^i{}^1_1{}^{a_1}_{b_3}\left({}^0_0\right)^i{}^1_1 \cdots \left({}^0_0\right)^i{}^1_1{}^{a_{n-1}}_{a_n}. \]
This is clearly doable with a local rule that checks its surroundings for the marker patterns, and the resulting CA are reversible with $c_i \rightarrow \ID$. Now given a word $w \in \Sigma^r$ and $t \in \N$, we append the word
\[ \left({}^0_0\right)^i{}^1_1{}^{a}_{b_1}\left({}^0_0\right)^i{}^1_1{}^{a}_{b_2}\left({}^0_0\right)^i{}^1_1 \cdots \left({}^0_0\right)^i{}^1_1{}^{a}_{b_{t-1}} \]
to $w$, producing in the spacetime diagram an $r \times t$-rectangle with $w$ on the bottom and $b_1b_2\cdots b_{t-1}$ on the lower track of the right border. Thus we have that $N_{c_i}(r,t) \geq 4^r \cdot 2^{t-1}$, so
\[ h(c_i) \geq \lim_{r \rightarrow \infty} \lim_{t \rightarrow \infty} \frac{r \log 4 + (t-1) \log 2}{t} = \log 2. \]
But the identity CA has zero entropy: $h(\ID)=0$.
\end{example}

We conclude this section with the following two open problems.

\begin{question}
Is $h_+ : \REV \to \Q$ continuous when $\REV$ has the pointwise topology?
\end{question}

\begin{conjecture}
Transitive CA are not closed in the uniform Bernoulli space.
\end{conjecture}

\section{Topologizing the CA by measuring preimages: the measure theoretic approach}
\label{sec:Measures}

\begin{definition}
Let $\mu$ be a measure on $\Sigma^\Z$. We define the pseudometric
\[ \delta^\mu (c, d) = \mu([D^c_d]_0). \]
In this way each measure $\mu$ gives a topology for $\CA$. We denote $c_i \stackrel{\mu}{\rightarrow} c$ if $c_i \rightarrow c$ holds in this topology.
\end{definition}

\begin{remark}
The pseudometric $\delta^\mu$ is a metric iff $\mu$ has full support.
\end{remark}

It is easy to see that $c_i \stackrel{\mu}{\rightarrow} c$ holds iff
\[ \mu(c^{-1}(C) \triangle c_i^{-1}(C)) \rightarrow 0 \]
holds for all clopen sets $C \subseteq \Sigma^\Z$. As we noted in the previous section, the pseudometric defined by the uniform Bernoulli measure coincides with that of the uniform Bernoulli space.

We now finish the proof of non-sequentiality in Theorem~\ref{thm:PointwiseProperties}. The result follows from the following connection between the difference set topologies and the pointwise topology. Note that since the pointwise topology is not sequential, Proposition~\ref{prop:Finer} does \emph{not} imply that it would be finer than the difference set topologies.

\begin{lemma}
Let $\mu$ be a measure on $\Sigma^\Z$, and let $X_1, X_2, \ldots$ be a countably infinite family of Borel sets with $\mu(X_i) > \epsilon$ for some $\epsilon > 0$. Let $X$ be the set of points $x \in \Sigma^\Z$ which appear in infinitely many of the $X_i$. Then $\mu(X) \geq \epsilon$.
\end{lemma}

\begin{proposition}
\label{prop:Finer}
Let $\mu$ be a measure in $\Sigma^\Z$. If $c_i \rightarrow c$ in the pointwise topology, then $\delta^\mu(c_i, c) \rightarrow 0$.
\end{proposition}

\begin{proof}
Assume the contrary. Without loss of generality, there then exists a positive $\epsilon$ with $\delta_\mu(c_i, c) > \epsilon$ for all $i$. Denoting $X_i = [D^c_{c_i}]_0$, the conditions of the above lemma are satisfied, and the set
\[ X = \{x \in \Sigma^\Z \;|\; c_i(x)_0 \neq c(x)_0 \mbox{ for infinitely many } i \} \]
has $\mu$-measure at least $\epsilon$. In particular, $X$ is nonempty, which contradicts the fact that $c_i \rightarrow c$ in the pointwise topology. 
\end{proof}

The following result shows that the difference set topologies are quite numerous and varied in form.

\begin{theorem}
If $\mu$ and $\nu$ are measures on $\Sigma^\Z$, $\mu \neq \nu$, and $\mu$ is ergodic, then $\mu$ and $\nu$ induce distinct topologies on $\CA$.
\end{theorem}

\begin{proof}
First, it is necessarily true that $\mu \not\gg \nu$, so we have a measurable set $B \subseteq \Sigma^\Z$ such that $\mu(B)=0$ and $\nu(B)=q>0$. By regularity, for every $i \in \N$ we find an open set $U_i \supseteq B$ such that $\mu(U_i) < i^{-1}$ and $\nu(U_i) < q + i^{-1}$. We can suppose that $U_{i+1} \subseteq U_{i}$ for all $i$.

For every $i$ we have an ascending chain of clopen sets $C_i^1 \subseteq C_i^2 \subseteq \ldots \subseteq U_i$ such that $U_i-C_i^n \searrow \emptyset$. From this it follows that $\nu(C_i^n) \rightarrow \nu(U_i)$. Thus, for all $i \in \N$, we can take a number $M_i \in \N$ such that $\nu(C_i) > q - i^{-1}$ holds for the set $C_i = C_i^{M_i}$. If $j \geq i$, we also have that
\[ \nu(C_i \cap C_j) \geq \nu(C_i)+\nu(C_j)-\nu(U_i) > q - 3i^{-1}. \]
Furthermore, it is clear that $\mu(C_i) \leq \mu(U_i) < i^{-1}$.

Let $W_i \subseteq \Sigma^{N(r_i)}$ be a finite set of words such that $C_i = [W_i]_0$. Define a sequence in $\CA$ as follows. The CA $c_i$ acts as the identity except on words of $W_i$, on which it applies the permutation $\tau$ to the central cell. Now $D^{\mbox{\scriptsize \ID}}_{c_i} = W_i$, so that
\[ \delta^\mu(\ID,c_i) = \mu(C_i) \rightarrow 0 \]
but
\[ \delta^\nu(\ID,c_i) = \nu(C_i) \rightarrow q > 0. \]
Thus $c_i \stackrel{\mu}{\rightarrow} \ID$ holds, but $c_i \stackrel{\nu}{\rightarrow} \ID$ does not, which proves the claim. 
\end{proof}

\begin{corollary}
All Bernoulli measures induce distinct topologies on $\CA$. No other shift invariant measure can induce these topologies.
\end{corollary}

We now reprove most of the combinatorial results of Section~\ref{sec:Combinatorics} in the measure theoretic framework. In some cases, the proofs become somewhat shorter and cleaner.

\begin{definition}
Let $\mu$ be a measure on $\Sigma^\Z$. Then $\PRES_\mu$ denotes the set of $\mu$-preserving CA, that is,
\[ \PRES_\mu = \{ c \in \CA \;|\; \mu(c^{-1}(B)) = \mu(B) \text{ for all } \mu \text{-measurable } B \subseteq \Sigma^\Z \}. \]
\end{definition}

The following theorems generalize the results for the uniform Bernoulli measure $\mu$, and are proved using the same ideas. Note that in this case, $\SUR = \PRES_\mu$.

\begin{theorem}
Let $\mu$ be a measure on $\Sigma^\Z$. The restriction of $(\circ)$ to $\CA \times \PRES_\mu$ is continuous with respect to $\delta^\mu$.
\end{theorem}

\begin{proof}
Let $c_i \in \CA$ and $d_i \in \PRES_\mu$ for all $i$, and $c_i \stackrel{\mu}{\rightarrow} c$, $d_i \stackrel{\mu}{\rightarrow} d$. We will prove that $c_i \circ d_i \stackrel{\mu}{\rightarrow} c \circ d$.

For that, let $C \subseteq \Sigma^\Z$ be a clopen set. Denote
\[ A_i = d_i^{-1}(c_i^{-1}(C)) \triangle d_i^{-1}(c^{-1}(C)), \]
the $d_i$-preimage of the set where $c_i$ and $c$ differ w.r.t. $C$, and
\[ B_i = d^{-1}(c^{-1}(C)) \triangle d_i^{-1}(c^{-1}(C)), \]
the set where $d_i$ and $d$ differ w.r.t $c^{-1}(C)$. It is clear that
\[ d_i^{-1}(c_i^{-1}(C)) \triangle d^{-1}(c^{-1}(C)) \subseteq A_i \cup B_i. \]

Since $d_i$ preserves $\mu$ and $c^{-1}(C)$ is a clopen set, we have that
\[ \mu(A_i) = \mu(c_i^{-1}(C) \triangle c^{-1}(C)) \rightarrow 0 \]
and
\[ \mu(B_i) = \mu(d^{-1}(c^{-1}(C)) \triangle d_i^{-1}(c^{-1}(C))) \rightarrow 0, \]
from which the claim then follows. 
\end{proof}

\begin{theorem}
If $\mu$ is a measure and $d \in \CA$, then $c \mapsto d \circ c$ is continuous with respect to $\delta^\mu$.
\end{theorem}

\begin{proof}
Let $c_i \stackrel{\mu}{\rightarrow} c$ and $C$ a clopen set. Since $d^{-1}(C)$ is also clopen, we have
\begin{align*}
& \phantom{M} \mu((d \circ c_i)^{-1}(C) \triangle (d \circ c)^{-1}(C)) \\
& = \mu(c_i^{-1}(d^{-1}(C)) \triangle c^{-1}(d^{-1}(C))) \rightarrow 0.
\end{align*} 
\end{proof}

\begin{theorem}
If $\mu$ is a measure, then $(\cdot)^{-1}$ is continuous on $\PRES_\mu \cap \REV$.
\end{theorem}

\begin{proof}
Let $c_i \stackrel{\mu}{\rightarrow} c$ in $\PRES_\mu \cap \REV$, and $C$ a clopen set. Clearly, also $c_i^{-1} \in \PRES_\mu$ for all $i$. Denote $D=c(C)$, which is also clopen. Now
\begin{align*}
& \phantom{M} \mu((c_i^{-1})^{-1}(C) \triangle (c^{-1})^{-1}(C)) \\
&= \mu(c_i(C) \triangle c(C)) \\
&= \mu(c_i(c^{-1}(D)) \triangle D) \\
&= \mu(c^{-1}(D) \triangle c_i^{-1}(D)) \rightarrow 0.
\end{align*} 
\end{proof}

\begin{theorem}
Let $\mu$ be a measure. Then $\PRES_\mu$ is closed in $\CA_\mu$.
\end{theorem}

\begin{proof}
Let $c \notin \PRES_\mu$. Now we have a measurable set $B$ such that $\mu(c^{-1}(B)) > \mu(B)$. Since $\mu$ is regular, we can find an open set $U \supseteq B$ such that $\mu(U) < \mu(c^{-1}(B)) \leq \mu(c^{-1}(U))$. Now $U$ can in turn be approximated with clopen sets $C_1 \subseteq C_2 \subseteq \ldots \subseteq U$ such that $U-C_j \searrow \emptyset$, which implies $\mu(C_j) \rightarrow \mu(U)$ and $\mu(c^{-1}(C_j)) \rightarrow \mu(c^{-1}(U))$. Take $j$ large enough that $\mu(C_j) \leq \mu(U) < \mu(c^{-1}(C_j))$ and denote $C=C_j$.

Suppose then that $c_i \in \PRES_\mu$ for all $i \in \N$, and $c_i \stackrel{\mu}{\rightarrow} c$. Now in particular $\mu(c_i^{-1}(C)) \rightarrow \mu(c^{-1}(C)) > \mu(C)$. But this contradicts the assumption $c_i \in \PRES_\mu$. 
\end{proof}

Next, we show how the difference set spaces can also be obtained by integrating different distance functions over the space $\Sigma^\Z$. In the case of the Besicovitch pseudometric (defined below), we have an exact correspondence with the respective $\delta^\mu$.

\begin{theorem}
Let $\mu$ be a shift invariant measure, and define, for all $c,d \in \CA$,
\[ \delta_C^\mu(c,d) = \int_{\Sigma^\Z} d_C(c(x),d(x)) \; d\mu(x). \]
Then $\delta^\mu$ and $\delta_C^\mu$ are uniformly equivalent pseudometrics.
\end{theorem}

\begin{proof}
First, note that
\[ \delta_C^\mu(c,d) \geq \int_{[D^c_d]_0} d_C(c(x),d(x)) \; d\mu(x) \geq \mu([D^c_d]_0). \]
Second, we have
\[ \delta_C^\mu(c,d) = \int_{\Sigma^\Z} \sum_{c(x)_i \neq d(x)_i} 2^{-|i|} \; d\mu(x) \leq \sum_{i \in \Z} \mu([D^c_d]_i) 2^{-|i|} = 3\mu([D^c_d]_0). \] 
\end{proof}

\begin{definition}
Let $x,y \in \Sigma^\Z$. The \emph{Besicovitch distance} of $x$ and $y$ is
\[ d_B(x,y) = \limsup_{n \rightarrow \infty} \frac{\abs{\{ \abs{i} \leq n \;|\; x_i \neq y_i \}}}{2n+1}, \]
that is, the limit superior of the density of their differences as the size of the observation window grows without bound.
\end{definition}

The Besicovitch distance is a pseudometric and defines a topology on $\Sigma^\Z$ different from the usual Cantor topology. The following theorem gives a connection between the pseudometrics $\delta^\mu$ and $d_B$.

\begin{theorem}
Let $\mu$ be a shift invariant measure, and define, for all $c,d \in \CA$,
\[ \delta_B^\mu(c,d) = \int_{\Sigma^\Z} d_B(c(x),d(x)) \; d\mu(x). \]
Then we have $\delta^\mu = \delta_B^\mu$.
\end{theorem}

\begin{proof}
Let $c,d \in \CA$ with common radius $r$. Now the function $\textbf{1}_{[D^c_d]_0}$ is $\mu$-integrable, so by the ergodic theorem we have that
\begin{align*}
\textbf{1}_{[D^c_d]_0}^*(x) =& \lim_{n \rightarrow \infty} \frac{1}{2n+1} \sum_{i=-n}^n \textbf{1}_{[D^c_d]_i}(x) \\
=& \lim_{n \rightarrow \infty} \frac{\abs{\{ i \in N(n) \;|\; x_{N(r)+i} \in D^c_d \}}}{2n+1}
\end{align*}
is defined for $\mu$-almost all $x$, and
\[
\int_{\Sigma^\Z} \textbf{1}_{[D^c_d]_0}^* d\mu = \int_{\Sigma^\Z} \textbf{1}_{[D^c_d]_0} d\mu = \mu([D^c_d]_0).
\]
The claim directly follows, since now $\textbf{1}_{[D^c_d]_0}^*(x) = d_B(c(x),d(x))$ for $\mu$-almost all $x \in \Sigma^\Z$. 
\end{proof}

\section{Future Work}
\label{sec:Future}

In the future, it would be interesting to consider connections between dynamical notions such as transitivity, mixingness, sensitivity and entropy by proving that under some additional natural constraints on the topologies considered, closedness of sets of cellular automata with a (possibly parametrized) dynamical property, and continuity of certain dynamical invariants automatically imply other closedness and continuity results in any topology. Such results, if any can be proved, could perhaps imply (or inspire) interesting new connections outside our topology framework between dynamical notions. We are also interested in whether entropy can be made continuous with a natural topology based on the long-term behavior of cellular automata.

We would also like to study the \emph{sequentialization} of the pointwise topology, the topology whose closed sets are exactly the sequentially closed sets of the pointwise topology. This is easily seen to define a sequential topology. This topology is finer than the pointwise topology and any of the difference set topologies, and both $\REV$ and $\SUR$ are closed sets (by Proposition~\ref{prop:REVYesClosed} and Theorem~\ref{thm:SURClosed}, respectively). However, it can be shown that the space is still perfect.

\bibliographystyle{eptcs}
\bibliography{bib}{}

\def\ocirc#1{\ifmmode\setbox0=\hbox{$#1$}\dimen0=\ht0 \advance\dimen0
  by1pt\rlap{\hbox to\wd0{\hss\raise\dimen0
  \hbox{\hskip.2em$\scriptscriptstyle\circ$}\hss}}#1\else {\accent"17 #1}\fi}
\begin{thebibliography}{10}
\providecommand{\bibitemdeclare}[2]{}
\providecommand{\surnamestart}{}
\providecommand{\surnameend}{}
\providecommand{\urlprefix}{Available at }
\providecommand{\url}[1]{\texttt{#1}}
\providecommand{\href}[2]{\texttt{#2}}
\providecommand{\urlalt}[2]{\href{#1}{#2}}
\providecommand{\doi}[1]{doi:\urlalt{http://dx.doi.org/#1}{#1}}
\providecommand{\bibinfo}[2]{#2}

\bibitemdeclare{article}{BoFrKi90}
\bibitem{BoFrKi90}
\bibinfo{author}{Mike \surnamestart Boyle\surnameend}, \bibinfo{author}{John
  \surnamestart Franks\surnameend} \& \bibinfo{author}{Bruce \surnamestart
  Kitchens\surnameend} (\bibinfo{year}{1990}):
  \emph{\bibinfo{title}{Automorphisms of one-sided subshifts of finite type}}.
\newblock {\sl \bibinfo{journal}{Ergodic Theory Dynam. Systems}}
  \bibinfo{volume}{10}(\bibinfo{number}{3}), pp. \bibinfo{pages}{421--449},
  \doi{10.1017/S0143385700005678}.

\bibitemdeclare{article}{BoLiRu88}
\bibitem{BoLiRu88}
\bibinfo{author}{Mike \surnamestart Boyle\surnameend}, \bibinfo{author}{Douglas
  \surnamestart Lind\surnameend} \& \bibinfo{author}{Daniel \surnamestart
  Rudolph\surnameend} (\bibinfo{year}{1988}): \emph{\bibinfo{title}{The
  Automorphism Group of a Shift of Finite Type}}.
\newblock {\sl \bibinfo{journal}{Transactions of the American Mathematical
  Society}} \bibinfo{volume}{306}(\bibinfo{number}{1}), pp. \bibinfo{pages}{pp.
  71--114}.
\newblock \urlprefix\url{http://www.jstor.org/stable/2000831}.

\bibitemdeclare{book}{DeGrSi76}
\bibitem{DeGrSi76}
\bibinfo{author}{Manfred \surnamestart Denker\surnameend},
  \bibinfo{author}{Christian \surnamestart Grillenberger\surnameend} \&
  \bibinfo{author}{Karl \surnamestart Sigmund\surnameend}
  (\bibinfo{year}{1976}): \emph{\bibinfo{title}{Ergodic theory on compact
  spaces}}.
\newblock \bibinfo{series}{Lecture Notes in Mathematics, Vol. 527},
  \bibinfo{publisher}{Springer-Verlag}, \bibinfo{address}{Berlin}.

\bibitemdeclare{article}{Fu55}
\bibitem{Fu55}
\bibinfo{author}{Harry \surnamestart Furstenberg\surnameend}
  (\bibinfo{year}{1955}): \emph{\bibinfo{title}{On the infinitude of primes}}.
\newblock {\sl \bibinfo{journal}{Amer. Math. Monthly}} \bibinfo{volume}{62}, p.
  \bibinfo{pages}{353}.

\bibitemdeclare{article}{He69}
\bibitem{He69}
\bibinfo{author}{G.~A. \surnamestart Hedlund\surnameend}
  (\bibinfo{year}{1969}): \emph{\bibinfo{title}{Endomorphisms and automorphisms
  of the shift dynamical system}}.
\newblock {\sl \bibinfo{journal}{Math. Systems Theory}} \bibinfo{volume}{3},
  pp. \bibinfo{pages}{320--375}.

\bibitemdeclare{article}{Ho10}
\bibitem{Ho10}
\bibinfo{author}{Michael \surnamestart Hochman\surnameend}
  (\bibinfo{year}{2010}): \emph{\bibinfo{title}{On the automorphism groups of
  multidimensional shifts of finite type}}.
\newblock {\sl \bibinfo{journal}{Ergodic Theory Dynam. Systems}}
  \bibinfo{volume}{30}(\bibinfo{number}{3}), pp. \bibinfo{pages}{809--840},
  \doi{10.1017/S0143385709000248}.

\bibitemdeclare{article}{Ka96}
\bibitem{Ka96}
\bibinfo{author}{J.~\surnamestart Kari\surnameend} (\bibinfo{year}{1996}):
  \emph{\bibinfo{title}{Representation of reversible cellular automata with
  block permutations}}.
\newblock {\sl \bibinfo{journal}{Theory of Computing Systems}}
  \bibinfo{volume}{29}, pp. \bibinfo{pages}{47--61}.
\newblock \urlprefix\url{http://dx.doi.org/10.1007/BF01201813}.
\newblock \bibinfo{note}{10.1007/BF01201813}.

\bibitemdeclare{article}{Ka90}
\bibitem{Ka90}
\bibinfo{author}{Jarkko \surnamestart Kari\surnameend} (\bibinfo{year}{1990}):
  \emph{\bibinfo{title}{Reversibility of 2D cellular automata is undecidable}}.
\newblock {\sl \bibinfo{journal}{Physica D: Nonlinear Phenomena}}
  \bibinfo{volume}{45}(\bibinfo{number}{1–3}), pp. \bibinfo{pages}{379 --
  385}, \doi{10.1016/0167-2789(90)90195-U}.
\newblock
  \urlprefix\url{http://www.sciencedirect.com/science/article/pii/016727899090%
195U}.

\bibitemdeclare{article}{Ka05}
\bibitem{Ka05}
\bibinfo{author}{Jarkko \surnamestart Kari\surnameend} (\bibinfo{year}{2005}):
  \emph{\bibinfo{title}{Theory of cellular automata: a survey}}.
\newblock {\sl \bibinfo{journal}{Theoret. Comput. Sci.}}
  \bibinfo{volume}{334}(\bibinfo{number}{1-3}), pp. \bibinfo{pages}{3--33},
  \doi{10.1016/j.tcs.2004.11.021}.

\bibitemdeclare{inproceedings}{KaVaZe09}
\bibitem{KaVaZe09}
\bibinfo{author}{Jarkko \surnamestart Kari\surnameend}, \bibinfo{author}{Pascal
  \surnamestart Vanier\surnameend} \& \bibinfo{author}{Thomas \surnamestart
  Zeume\surnameend} (\bibinfo{year}{2009}): \emph{\bibinfo{title}{Bounds on
  Non-surjective Cellular Automata}}.
\newblock In: {\sl \bibinfo{booktitle}{Proceedings of the 34th International
  Symposium on Mathematical Foundations of Computer Science 2009}},
  \bibinfo{series}{MFCS '09}, \bibinfo{publisher}{Springer-Verlag},
  \bibinfo{address}{Berlin, Heidelberg}, pp. \bibinfo{pages}{439--450},
  \doi{10.1007/978-3-642-03816-7\_38}.

\bibitemdeclare{article}{SaTo12}
\bibitem{SaTo12}
\bibinfo{author}{V.~\surnamestart {Salo}\surnameend} \&
  \bibinfo{author}{I.~\surnamestart {T{\"o}rm{\"a}}\surnameend}
  (\bibinfo{year}{2012}): \emph{\bibinfo{title}{{Geometry and Dynamics of the
  Besicovitch and Weyl Spaces}}}.
\newblock {\sl \bibinfo{journal}{ArXiv e-prints}}.

\bibitemdeclare{book}{SeSt95}
\bibitem{SeSt95}
\bibinfo{author}{Lynn~A. \surnamestart Steen\surnameend} \&
  \bibinfo{author}{Arthur~J. \surnamestart Seebach\surnameend}
  (\bibinfo{year}{1995}): \emph{\bibinfo{title}{{Counterexamples in
  Topology}}}.
\newblock \bibinfo{publisher}{{Dover Publications}}.
\newblock
  \urlprefix\url{http://www.amazon.com/exec/obidos/redirect?tag=citeulike07-20%
\&path=ASIN/048668735X}.

\end{thebibliography}

\end{document}